\documentclass[letterpaper, times, mathptm,10 pt, conference]{ieeeconf}  

\IEEEoverridecommandlockouts                              

\usepackage{amsmath,amssymb,color,url,graphicx,algorithm,algorithmic,subcaption}
\usepackage{xurl}

\newcommand{\setN}{\mathcal{N}}
\newcommand{\setL}{\mathcal{L}}
\newcommand{\setS}{\mathcal{S}}
\newcommand{\setA}{\mathcal{A}}

\newtheorem{definition}{Definition}
\newtheorem{proposition}{Proposition}

\newcommand{\revise}[1]{{\color{black}#1}}

\title{\LARGE \bf
Routing and charging game in ride-hailing service with electric vehicles
}

\author{Kenan Zhang$^{1}$ and John Lygeros$^{2}$
\thanks{$^{1}$School of Architecture, Civil and Environmental Engineering, EPFL, Switzerland ({\tt\small kenan.zhang@epfl.ch}).}
\thanks{$^{2}$Department of Information Technology and Electrical Engineering, ETH Z\"urich, Z\"urich, Switzerland ({\tt\small jlygeros@ethz.ch}).}
}

\begin{document}

\maketitle
\thispagestyle{empty}
\pagestyle{empty}

\begin{abstract}
This paper studies the routing and charging behaviors of electric vehicles in a competitive ride-hailing market. 
When the vehicles are idle, they can choose whether to continue cruising to search for passengers, or move a charging station to recharge. The behaviors of individual vehicles are then modeled by a Markov decision process (MDP). The state transitions in the MDP model, however, depend on the aggregate vehicle flows both in service zones and at charging stations. Accordingly, the value function of each vehicle is determined by the collective behaviors of all vehicles. With the assumption of the large population, we formulate the collective routing and charging behaviors as a mean-field Markov game. We characterize the equilibrium of such a game, prove its existence, and numerically show that the competition among vehicles leads to ``inefficient congestion" both in service zones and at charging stations. 
\end{abstract}

\section{INTRODUCTION}
The past decade has witnessed rapid growth in the electric vehicle (EV) market. 
\revise{In 2021, global EV sales increased by 109\%, more than doubled in 2020,}
and maintained a considerable growth rate of 55\% in 2022 \cite{ev2022ev}. 
Governments and local authorities play a critical role in promoting the adoption of EVs. In 2021, the US announced its target of 50\% share of EVs in new car sales by 2023, along with a funding package of \$7.5 billion for charging infrastructure \cite{bibra2022global}. Early this year, EU gave the final approval to end the sale of fossil fuel vehicles by 2035 \cite{eu2022zero}. 

Envisioning the wide adoption of EVs, \revise{recent research has been} 
devoted to modeling their coupled charging and routing behaviors in the transportation system \cite{he2014network}. The early studies mostly focus on the planning of charging stations \cite{wang2013traffic,he2013optimal}. On the other hand, studies at the operational level often take the perspective of grid operators. To induce a desirable EV charging pattern, the grid operator either directly optimizes the charging price \cite{qian2020enhanced}, or indirectly designs the power generation plan, which then gives the charging price via the locational marginal pricing (LMP) mechanism \cite{alizadeh2016optimal,manshadi2017wireless}. 
\revise{All these studies} 
target private EVs and assume they must recharge at least once during their trips. However, the current battery capacity of most EV models is more than enough for a single trip within the city. A recent report also shows that most private EVs charge at home or at work~\cite{virta}. 
In contrast, routing and charging are indeed important decisions for EVs in ride-hailing services (e.g., taxi and ride-sourcing). 
These vehicles usually travel for a much longer distance every day compared to private vehicles~\cite{schaller2017empty}. 
Moreover, ride-hailing vehicles are driven by travel demand and move across different regions in the city. Hence, they are more likely to charge at public charging stations. 
This routing and charging problem is expected to be more prevailing given the increasing popularity of shared mobility services \cite{mckinsey2021shared} and the fact that major players (e.g., Uber and Lyft) are electrifying their fleets \cite{Khurana2021ride}. 
Research on this topic, however, is still rare in the literature, with most existing studies assuming that the vehicle routing and charging are centrally controlled by the service operator \cite{rossi2019interaction,jamshidi2021dynamic,santos2023space}. Two exceptions are \cite{qin2020piggyback} and \cite{ding2022optimal}, both of which consider the case where EVs provide a transport service by ride-hailing and an energy service by vehicle-to-grid. Yet, both models are static and ignore charging behaviors.

In this paper, we \revise{extend our previous work \cite{zhang2023ride} to study} the electric ride-hailing vehicle routing problem (eRIVER) in a spatiotemporal market. Specifically, we assume that each vehicle makes routing and charging decisions to maximize its own profit over a single-day operation. We model each vehicle's behaviors by a Markov decision process (MDP) and carefully design the state transition functions to reflect the physical interactions among vehicles in each service zones and at each charging station. We then show that the value function of each vehicle depends on the collective behaviors of all vehicles. When the fleet size is sufficiently large, the impact of each vehicle on the aggregate vehicle flows is negligible. This motivates us to formulate collective vehicle behaviors as a mean-field game. 
In the remainder of this paper, we first present the eRIVER model, define the equilibrium and discuss its existence, then conduct numerical experiments to explore the equilibrium vehicle flows in a stylized network market.

\section{THE ERIVER MODEL}
Consider a spatiotemporal ride-hailing market that is discretized into $N$ zones containing $L$ charging stations, and $T$ time steps with equal length $\Delta$. 
For simplicity, we assume all charging stations are homogeneous with capacity $C$ and charging efficiency $e$. 
Let $\setN$ and $\setL$ be the sets of service zones and charging stations, respectively. The travel time from zone $i\in\setN$ to zone $j\in\setN$ (station $l\in\setL$) is denoted by $\tau^t_{ij}$ ($\tau^t_{il}$) and is measured in units of $\Delta$. Similarly, the vehicle charging time is measured in units of $\Delta$ as well. 
Consider a fleet of $M$ self-interested and homogeneous electric vehicles \revise{operating} in the market, each with battery capacity $B$. To be consistent with other variables, the state of charge (SOC) is also discretized and measured in units of $\Delta$ and the battery consumption rate is $\xi$ per time step. 
We further assume vehicles \revise{do not} leave the charging station until they are fully charged. Hence, a vehicle that starts charging at SOC $b=0,\dots,B-1$ will leave the charging station after $\hat{\tau}_b = \lceil(B-b)/e \rceil$ time steps. 

\revise{While the fleet size $M$ is finite, we consider it is sufficiently large such that the aggregate behaviors of vehicles can be represented by continuous flows. This is a common assumption in transportation research~\cite{sheffi1985urban} and closely related to mean-field game~\cite{lasry2007mean}, the game theoretical framework adopted in eRIVER. Specifically, two vehicle flows are defined as follows:
}
\begin{itemize}
    \item $y^t_{i,b}$: idle vehicles in zone $i$ at time $t$ with SOC $b$.
    \item $z^t_{l,b}$: vehicles arriving at station $l$ at time $t$ with SOC $b$.
\end{itemize}


\subsection{Matching in zones} 
Let $q^t_i$ be the number of passengers arriving in zone $i$ at time $t$. Assume passengers can only be matched with vehicles in the same zone, and they leave the ride-hailing market for an alternative travel mode after one period of wait.  
\revise{Meanwhile, each idle vehicle can only be matched with one passenger (i.e., no ride-pooling).}
Then, for each idle vehicle in the same zone, the probability of successfully picking up a passenger after one period of search is given by 
\begin{align}
    m^t_i = f(q^t_i, y^t_i; \theta^t_i),\label{eq:meet-prob}
\end{align}
where $y^t_i = \sum_{b>0} y^t_{i,b}$ and $\theta^t_i$ denotes the time and location-specific parameter (e.g., road density, travel speed). In what follows, $m^t_i$ is referred to as the \emph{meeting probability}. To ensure the model is realistic,  we assume that the function $f$ increases with passenger demand ($\mathrm{d}f/\mathrm{d} q \geq 0$) and decreases with vehicle supply  ($\mathrm{d} f/\mathrm{d} y \leq 0$). 

Accordingly, after one period of search, $m^t_i$ of the idle vehicles in zone $i$ successfully pick up a passenger. They will start delivering their trips from the next time step and become idle again after several time steps (depending on the trip duration). The other vehicles are going to make a new routing and charging decision in zone $i$ at time $t+1$.

\subsection{Charging at stations}
When a vehicle arrives at a charging station, some vehicles could already be there, either charging or waiting to charge.
Besides, there could be a group of other vehicles arriving at the same time. Hence, the waiting time till the vehicle starts charging\revise{, denoted by $\omega$,} is probabilistic depending on both vehicles arriving at time $t$ and those arriving earlier, which are collectively represented by $\mathbf{z}_{\leq t,l}=\{z^{t'}_{l,b}\}_{t'\leq t, b}$. 
\revise{The probability mass function of the waiting time $\omega$ for a vehicle arriving at station $l$ at time $t$ is characterized as}
\begin{align}
    w^t_l(\omega) = g(\mathbf{z}_{\leq t, l}, \omega).\label{eq:wait-dist}
\end{align}
By definition, $w^t_l(\omega)\in[0,1],\forall \omega\in\mathbb{N}$ and $\sum_\omega w^t_l(\omega) = 1$.
Accordingly, a newly arrived vehicle with SOC $b$ first waits for $\omega$ time steps with probability $w^t_l(\omega)$ and then charges for $\hat{\tau}_b$ time steps before leaving with a full battery.


\subsection{MDP for individual vehicle}
Now, we are ready to specify the finite-horizon Markov decision process (MDP) for each vehicle. 

\subsubsection{State $s\in\setS$}  
Given by the time step $t$, location $k\in\setN\cup\setL\cup\{0\}$, and SOC $b$. Here, location $0$ denotes an offline state when the vehicle runs out of battery.

\subsubsection{Action $a\in\setA$} 
The set of feasible actions depends on the vehicle's state. 
When the vehicle is in zone $i$, it can choose either to continue cruising or move to a charging station. Hence, the action set is given by $\setN_i\cup\setL$, where $\setN_i$ includes zone $i$ itself and its neighbor zones. 
When the vehicle is about to leave station $l$, it can choose one of the neighbor zones for cruising and thus the action belongs to the set of neighbor zones of station $l$, denoted by $\setN_l$. 
Once the vehicle runs out of battery, it can do nothing but stay offline.

\subsubsection{State transition $P:\setS\times\setA\times\setS\to[0,1]$}
The probability of transitions between every pair of states under each action.
The meeting probability $m^t_i$ and the waiting time distribution $w^t_l$ play a critical role here. For instance, a transition from state $s=(t, i, b)$ to state $s'=(t+1,j,b-1)$ means the vehicle starting from zone $i$ fails to find a passenger in zone $j$ and thus the transition probability is given by $P(s'|s,a)=1-m^t_j$. If the same vehicle picks up a passenger traveling from zone $j$ to zone $k$ and the vehicle has sufficient battery to finish the trip, then $s'=(t+1+\tau_{jk},k,b-1-\tau_{jk})$ and $P(s'|s,a) = m^t_j\alpha^t_{jk}$, where $\alpha^t_{jk}$ denote the fraction of passengers in zone $j$ traveling to zone $k$ at time $t$. After making a charging decision, the vehicle first travels to the station and then waits for a while if there is a queue. If the waiting time is $\omega$, the next state becomes $s'=(t+\tau_{il}+\omega+\hat{\tau}_{b-\tau_{il}},l,B)$. The corresponding transition probability is $P(s'|s,a) = w^{t+\tau_{il}}_l(\omega)$. 

Note that we restrict the action set of cruising to the current and neighbor zones because idle ride-hailing vehicles continuously cruise before picking up a passenger. Even if they have a clear search target, the path can be decomposed into a sequence of local cruising destinations. For the same reason, we do not specify the cruising time between zones.
Besides, vehicles only need to may decisions when they are idle and when they just finish charging, thanks to the fixed travel time and changing efficiency. Hence, the state vector does not include vehicle operation status and the state transitions are not synchronized.

\subsubsection{Reward $r:\setS\times\setA\times\setS\to\mathbb{R}$} The immediate reward associated with each state transition.
Its value is non-zero in three cases: (i) picking up a passenger, (ii) starting to charge, (iii) being offline. The first case induces a positive reward $p_{jk}$ defined as the trip fare between zones $j$ and $k$, while the other two lead to a negative reward. The reward in (ii) is $c_l\hat{\tau}_b$ where $c_l$ is the charging price per unit of time. That in (iii) is an arbitrarily large penalty $\kappa$ for running out of battery during the day.

\subsubsection{Discount factor $\gamma\in(0,1]$}
express how much future rewards are taken into consideration in the current time step. In this study, $\gamma$ is set to be 1 because we study a single-day operation. The notation is thus omitted in the equations hereafter for simplicity.

The objective of each vehicle is to maximize its cumulative reward over time under a policy $\pi:\setS\to\Delta(\setA)$ and initial state distribution $\rho$, which is given by
\begin{align}
    &V_\rho(\pi|\mathbf{y},\mathbf{z}) = \nonumber\\
    &\mathbb{E}_{{s_0}\sim\rho}\left[\mathbb{E}_{a\sim\pi(\cdot|s), (s,a,s')\sim P(\cdot|\mathbf{y},\mathbf{z})}\left[\left.\sum_{(s,a,s')} r(s,a,s')\right|s_0\right]\right].\label{eq:value}
\end{align}
The key difference of the value function \eqref{eq:value} from a classic MDP is the conditioning on the aggregate vehicle flows $\mathbf{y}=\{y^t_{i,b}\}_{t,i,b}$ and $\mathbf{z}=\{z^t_{l,b}\}_{t,l,b}$. When these are fixed, one can solve for the optimal policy via dynamic programming. However, as will be shown in the next section, $\mathbf{y}$ and $\mathbf{z}$ are induced by the mean policy among the vehicles. Therefore, vehicles are not solving independent MDP problems but playing a mean-field Markov game.

\subsection{Mean-field equilibrium of eRIVER}
Since the vehicles are homogeneous and the fleet size is sufficiently large, we may use the mean policy among all vehicles to represent their collective routing strategies. Let $\delta(\pi)$ be the probability density function of policies among the vehicles and $\Omega$ be the set of feasible policies, the mean policy is given by
\begin{align}
    \bar{\pi} = \int_{\pi\in\Omega} \pi \delta(\pi)\mathrm{d}\pi. 
\end{align}
Note that vehicles do not necessarily share the same policy even though they are homogeneous. 
In what follows, we show how the vehicle flows $\mathbf{y}$ and $\mathbf{z}$ can be fully determined by $\bar{\pi}$ and the initial vehicle distribution $\rho$. To this end, we introduce $x^t_{k,b},k\in\setN\cup\setL$ to denote the vehicle flow in each zone and at each station before making the next routing and charging decision. Accordingly, we have
\begin{align}
    y^t_i &= \sum_{k\in\setN_i} \sum_{b>0} x^t_{k,b} \bar{\pi}(a=i|s=(t,k,b))\label{eq:y-fun}\\
    &\quad + \sum_{k\in\setL_i} x^t_{k,B} \bar{\pi}(a=i|s=(t,k,B)),\nonumber\\
    z^t_l &= \sum_{k\in\setN} \sum_{b\geq \tau_{kl}} x^{t-\tau_{kl}}_{k,b} \bar{\pi}(a=l|s=(t-\tau_{kl},k,b). \label{eq:z-fun}
\end{align}
On the other hand, $x^t_{k,b}$ can also be written as a function of $\mathbf{y}$ and $\mathbf{z}$. Specifically, for $k\in\setN$, 
\begin{align}
    x^t_{k,b} =  (1-m^{t-1}_k) y^{t-1}_{k,b+1}   + \sum^t_{t'=1}\sum_{i\in\setN(k,t,t')} \alpha^{t'}_{ik}m^{t'}_i y^{t'}_{i,b+\tau_{ik}}, \label{eq:x-func-zone}
\end{align}
where $\setN(k,t,t') =\{i\in\setN: t=t'+1+\tau_{ik}\}$. Likewise for $k\in\setL$, 
\begin{align}
    x^t_{k,B} = \sum_b \sum_{\omega} w^{t-\omega-\hat{\tau}_b}_k(\omega) z^{t-\omega-\hat{\tau}_b}_{k,b}.\label{eq:x-func-station}
\end{align}
Finally, the initial vehicle flows are determined by the initial state distribution, i.e., $x^0_{k,b} = M\rho(s=(0,k,b))$. 

Accordingly, we introduce a mapping $\mu$ such that $(\mathbf{y},\mathbf{z})=\mu(\bar{\pi})$ and define the mean-field equilibrium as follows.

\begin{definition}[Mean-field equilibrium in eRIVER]
A mean policy $\bar{\pi}^*$ is called a mean-field equilibrium (MFE) of eRIVER if it satisfies
\begin{align}
    \bar{\pi}^* \in \arg\max_\pi V_\rho(\pi|\mu(\bar{\pi}^*)). \label{eq:MFE}
\end{align}
\end{definition}

Classically, MFE for a Markov game is usually defined on a tuple of stationary policy and state (or state-action) distribution. See, for example, \cite{guo2019learning, light2022mean}. The notion of stationarity is required due to the infinite-horizon setting. Since eRIVER is modeled in a finite horizon, the mean-field distribution is determined by the mean policy. Besides, all information required for the routing problem is expressed by the aggregate flow $(\mathbf{y},\mathbf{z})$, which can also be seen as an integrated version of state-action distribution when normalized by the fleet size $M$. 


Under mild conditions, one can show that an MFE of eRIVER is guaranteed to exist. 

\begin{proposition}[Existence of equilibrium in eRIVER]\label{prop:exist}
If $f$ and $g$ are continuous in $(\mathbf{y},\mathbf{z})$, there exists at least one MFE for the eRIVER problem. 
\end{proposition}
\begin{proof}
Note that \eqref{eq:MFE} can be written as a fixed point $\bar{\pi}^*\in \phi(\bar{\pi}^*)$, where $\phi$ is a set-valued function and can be decomposed as $\phi=\psi\circ\mu$. Specifically, $\psi$ maps from the aggregate vehicle flow to the set of optimal policies and $\mu$ maps from a mean policy to the aggregate vehicle flows. As shown in \cite{zhang2023ride}, the existence of this fixed point can be proved by Kakutani's fixed point theorem~\cite[Theorem 8.6]{granas2003fixed}. The only non-trivial condition is that $\phi$ has a closed graph. In our setting (i.e., the feasible set of policies $\Omega$ is a compact set in a Hausdorff space), this is equivalent to proving $\phi$ is upper hemicontinuous.  
Following \cite{zhang2023ride}, we first show $\mu$ is single-valued and continuous with $\bar{\pi}$, and then prove $\psi$ is hemicontinuous. 

By~\eqref{eq:y-fun} and \eqref{eq:z-fun}, $\mathbf{y}$ and $\mathbf{z}$ are continuous in $\mathbf{x}=\{x^t_{k,b}\}_{t,k,b}$ and $\bar{\pi}$. In turn, \eqref{eq:x-func-zone} and \eqref{eq:x-func-station} suggest that $\mathbf{x}$ is continuous in $\mathbf{y}$ and $\mathbf{z}$ given that $f$ and $g$ are continuous functions of $(\mathbf{y},\mathbf{z})$. Therefore, $\mu$ is continuous in $\bar{\pi}$ and the mapping is single-valued. 
The remaining task is to show $\psi$ is hemicontinuous. This is done by invoking Berge's maximum theorem~\cite[Theorem 3.5]{hu2013handbook} with observations: (i) $\Omega$ is independent of the aggregate vehicle flows $(\mathbf{y}, \mathbf{z})$, and (ii) the value function \eqref{eq:value} is continuous due to continuous state transition $P(s'|a,s)$. The latter also results from the continuity of $f$ and $g$. 
\end{proof}


We note that the assumption imposed in Proposition~\ref{prop:exist} often holds. First, $f$ can be easily designed to be a continuous function of $\mathbf{y}$ (see Appendix~A). 
Establishing the continuity of waiting time distribution is more challenging, as it depends on a sequence of vehicle flows $\mathbf{z}_{\leq t, l}$. In Appendix~B, we derive a closed-form expression of $g$ and show it is indeed continuous in $\mathbf{z}$.

\section{NUMERICAL ANALYSIS}
\subsection{Settings and solution algorithm}
We analyze the equilibria of eRIVER on a stylized network, shown in Figure~\ref{fig:network}. The default values of exogenous model parameters are reported in Table~\ref{tab:default}.

\begin{figure}[htb]
    \centering
    \includegraphics[width=0.2\textwidth]{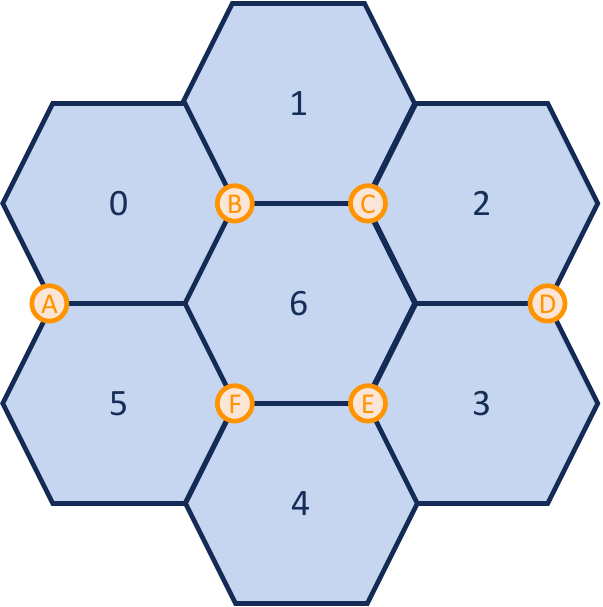}
    \caption{Stylized network with seven service zones $\setN=\{0,1,\dots,6\}$ and six charging stations $\setL=\{A,B,\dots,F\}$.}
    \label{fig:network}
\end{figure}

\begin{table}[htb]
\centering
\caption{Default values of exogenous variables.}\label{tab:default}
\begin{tabular}{|l|l|l||l|l|ll|}
\hline
Notation   & Unit     & Value & Notation     & Unit      & Value  & \\ \hline
$N$        &          & 7     & $\tau_{ij}$  & $\Delta$  & 1      & $j\in\setN_i$  \\ 
$L$        &          & 6     &              &           & 2      & $j\notin\setN_i$\\ 
$T $       &          & 12    & $\tau_{il}$  & $\Delta$  & 1      & $l\in\setL_i$    \\ 
$\Delta$   & hr       & 0.25  &              &           & 2      & $l\in\setL_{\setN_i}$    \\ 
$M$        & veh      & 500   &              &           & 3      & otherwise    \\ 
$B$        & $\Delta$ & 4     &   $p_{ij}$  &  \$       & 1      &   $\tau_{ij}=1$  \\ 
$C$        & veh      & 20    &      &          & 2      & $\tau_{ij}=2$    \\ 
$e$        & $\Delta$ & 3     & $c_l$  &  \$       & 0.5    &      $\l\in\setL$\\ 
$\xi$   & $\Delta$   & 1 &  $\kappa$    & \$        &  -10   &   \\ \hline
\multicolumn{7}{|l|}{Note: $\setL_i$ denotes the set of charging stations at the boundary of zone $i$,} \\
\multicolumn{7}{|l|}{and $\setL_{\setN_i}$ denotes the set of charging stations at the boundary of}\\
\multicolumn{7}{|l|}{neighbor zones of zone $i$.} \\ \hline
\end{tabular}
\end{table}

The meeting probability and charging waiting time distribution are specified according to Appendices A and B. 
All vehicles are idle at the beginning of the study horizon and, in the first set of experiments, they are evenly distributed over service zones with a full battery, i.e., $\rho(s_0=(0,i,B)) = 1/N,\;i\in\setN$.

We apply the Frank-Wolfe algorithm~\cite{frank1956algorithm} to compute the equilibrium. In each iteration, we first perform a forward propagation to load the vehicle flows using the current mean policy and then conduct a backward propagation to solve an optimal policy and update the mean policy with it. This iterative procedure is summarized in Algorithm~\ref{algo:equilibrium}. 

\begin{algorithm}[htb]
\caption{Solution algorithm for eRIVER}\label{algo:equilibrium}
\begin{algorithmic}[1]
\renewcommand{\algorithmicrequire}{\textbf{Input:} }
\renewcommand{\algorithmicensure}{\textbf{Output:}}
\REQUIRE Demand $\{q^t_i\}$; parameters in Tab.~\ref{tab:default}; gap threshold $\delta$
\ENSURE  Equilibrium policy $\bar{\pi}^*$
\\ Initiate random policy $\bar{\pi}^{(0)}$.
\FOR{$n = 0,1,\dots,$}
\STATE \textit{Forward propagation:} Load vehicle flows by \eqref{eq:y-fun} and \eqref{eq:z-fun} using policy $\bar{\pi}^{(n)}$.
\STATE \textit{Backward propagation:} Solve a policy $\hat{\pi}$ that maximizes \eqref{eq:value} by dynamic programming.
\STATE Update policy $\bar{\pi}^{(n+1)} = (1-\eta)\bar{\pi}^{(n)} + \eta \hat{\pi}$ with step size $\eta = 1/(n+1)$.
\STATE Compute gap $g= ||\bar{\pi}^{(n+1)} - \bar{\pi}^{(n)}||_1$.
\IF{$g<\delta$}
    \STATE \textbf{break}
\ENDIF
\ENDFOR
\RETURN $\bar{\pi}^* = \bar{\pi}^{(n)}$
\end{algorithmic}
\end{algorithm}

We compute the eRIVER equilibria for three demand profiles, where the origin-destination (OD) pattern is always assumed to be balanced, i.e., $\alpha^t_{ij}=1/N,\forall i,j,t$. 
\begin{itemize}
    \item \emph{Uniform}: invariant over time and space ($q^t_i=20,\forall t,i$);
    \item \emph{Peak/offpeak}: uniform over space with a temporal pattern shown in Figure~\ref{fig:q_s1};
    \item \emph{Central/peripheral}: uniform over time but concentrated in the central zone ($q^t_i=40$ for $i=6$ otherwise 15). 
\end{itemize}

\begin{figure}[htb]
    \centering
    \includegraphics[width=0.35\textwidth]{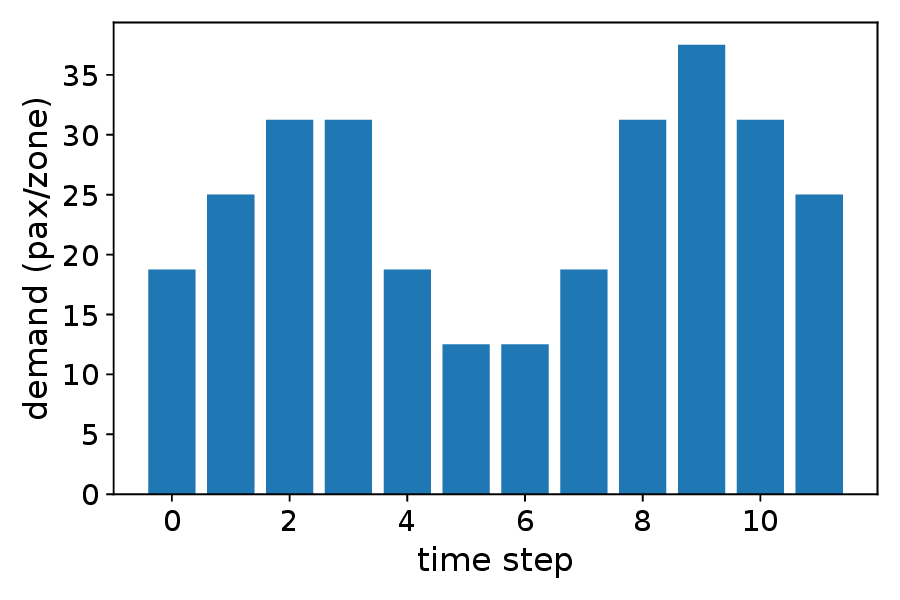}
    \caption{Temporal demand pattern in case Peak/offpeak.}
    \label{fig:q_s1}
\end{figure}

Figure~\ref{fig:gap_s3} illustrates the performance of Algorithm~\ref{algo:equilibrium} in the case of \emph{Peak/offpeak} demand; results in other cases are similar. The policy gap reduces to a magnitude of 10$^{-4}$ within 500 iterations. \revise{Yet, due to the diminishing step, the convergence is sublinear and the value gap (i.e., normalized difference between value function and Q-values with positive vehicle flows) stablized around 10$^{-3}$.}

\begin{figure}[htb]
    \centering
    \includegraphics[width=0.35\textwidth]{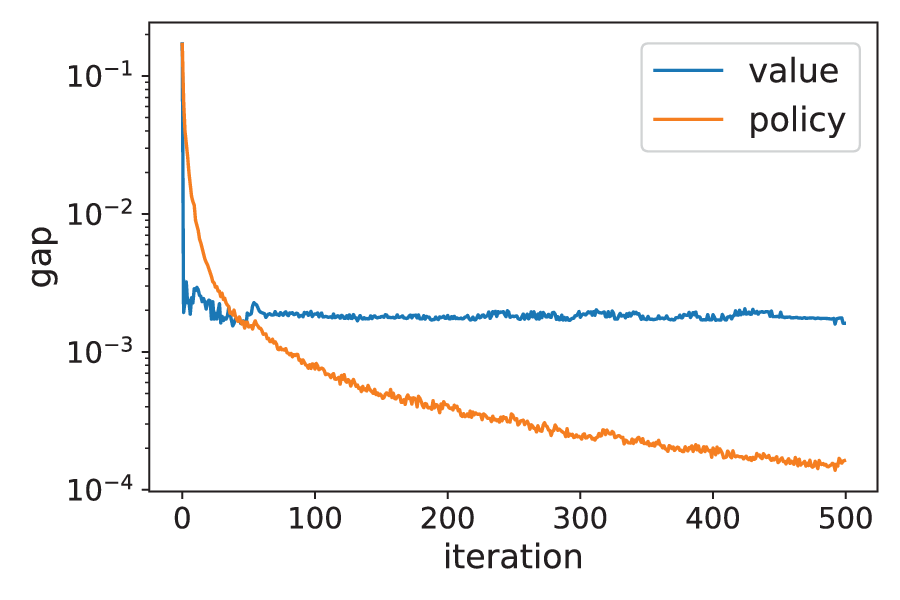}
    \caption{Gap over iteration in case Peak/offpeak.}
    \label{fig:gap_s3}
\end{figure}

\subsection{``Congestion'' in service zones and charging stations}
Figure~\ref{fig:congestion} plots equilibrium vehicle flows in service zones and at charging stations, normalized by their respective maximum values in each case. The darker the color, the larger the vehicle flows. In all three cases, vehicles show a strong preference for the central zone even though it has the same demand as other zones in the cases of \emph{Uniform} and \emph{Peak/offpeak} demand.
Besides, severe charging queues are observed at time $t=4$ because most vehicles arrive at the charging stations at time $t=3$ and $t=4$.

\begin{figure*}[htb]
    \centering
    \begin{subfigure}{\textwidth}
    \centering
    \includegraphics[width=\textwidth]{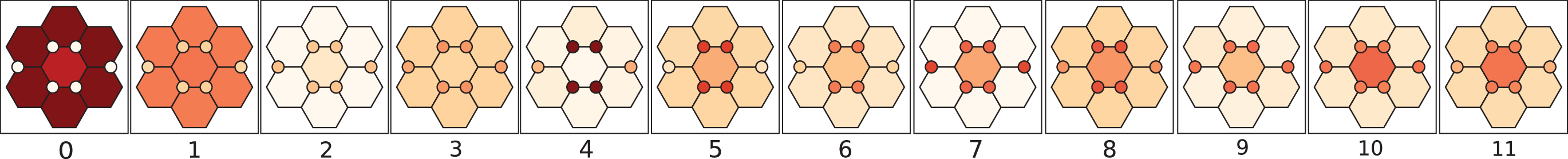}
    \caption{Uniform}
    \end{subfigure}
    \begin{subfigure}{\textwidth}
    \centering
    \includegraphics[width=\textwidth]{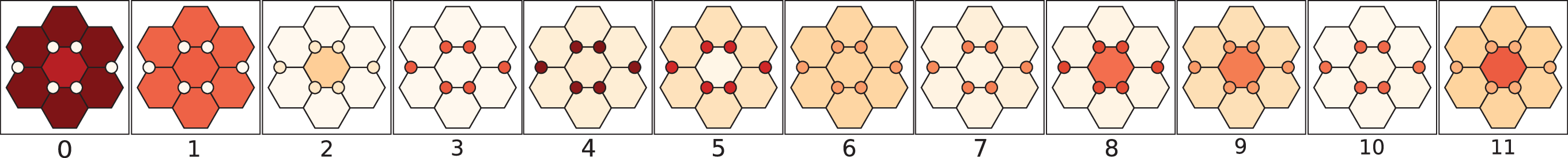}
    \caption{Peak/offpeak}
    \end{subfigure}
    \begin{subfigure}{\textwidth}
    \centering
    \includegraphics[width=\textwidth]{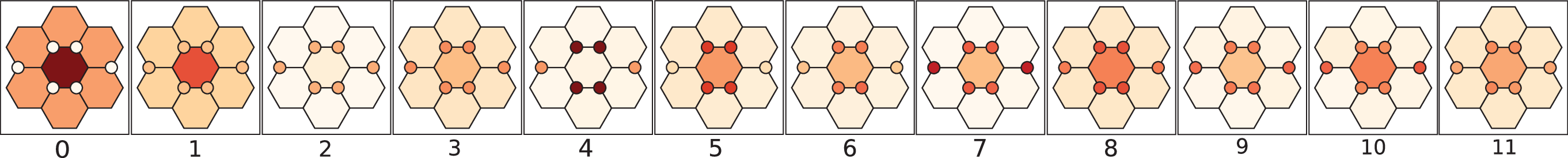}
    \caption{Central/peripheral}
    \end{subfigure}
    \caption{Normalized vehicle flows in service zones and at charging stations for the three demand patterns. The number under each panel indicates its time step.}
    \label{fig:congestion}
\end{figure*}

\begin{figure*}[htb]
    \centering
    \begin{subfigure}{0.3\textwidth}
        \centering
        \includegraphics[width=0.8\textwidth]{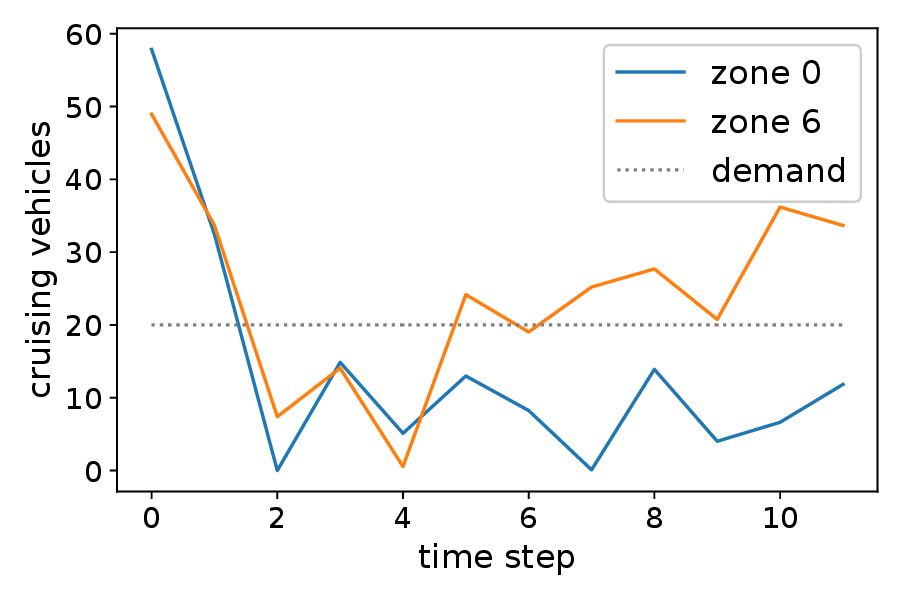}
        \caption{Uniform}
    \end{subfigure}
    \begin{subfigure}{0.3\textwidth}
        \centering
        \includegraphics[width=0.8\textwidth]{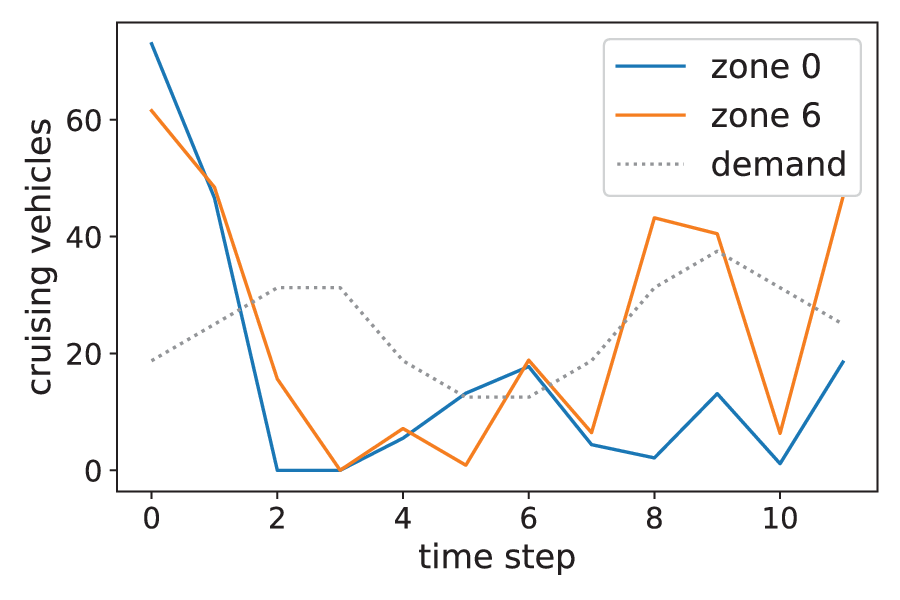}
        \caption{Peak/offpeak}
    \end{subfigure}
    \begin{subfigure}{0.3\textwidth}
        \centering
        \includegraphics[width=0.8\textwidth]{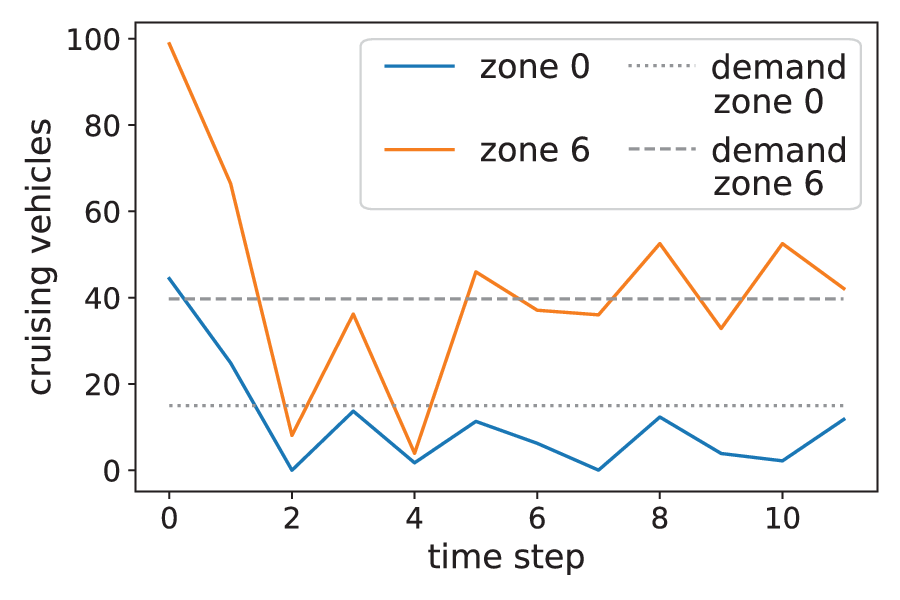}
        \caption{Central/peripheral}
    \end{subfigure}
    \caption{Cruising vehicles in central and peripheral zones.}
    \label{fig:cruise_veh}
\end{figure*}

\begin{figure*}[htb]
    \centering
    \begin{subfigure}{0.3\textwidth}
        \centering
        \includegraphics[width=0.8\textwidth]{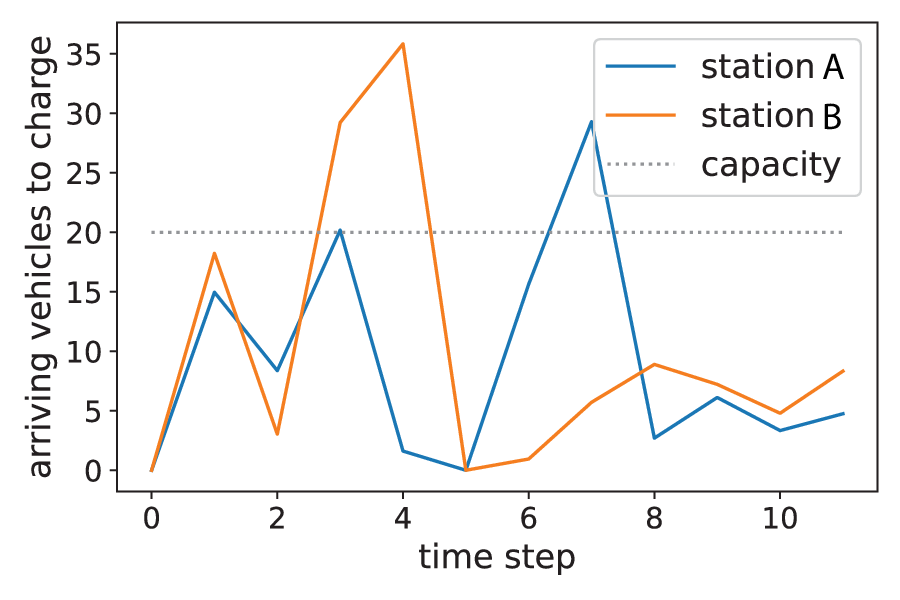}
        \caption{Uniform}
    \end{subfigure}
    \begin{subfigure}{0.3\textwidth}
        \centering
        \includegraphics[width=0.8\textwidth]{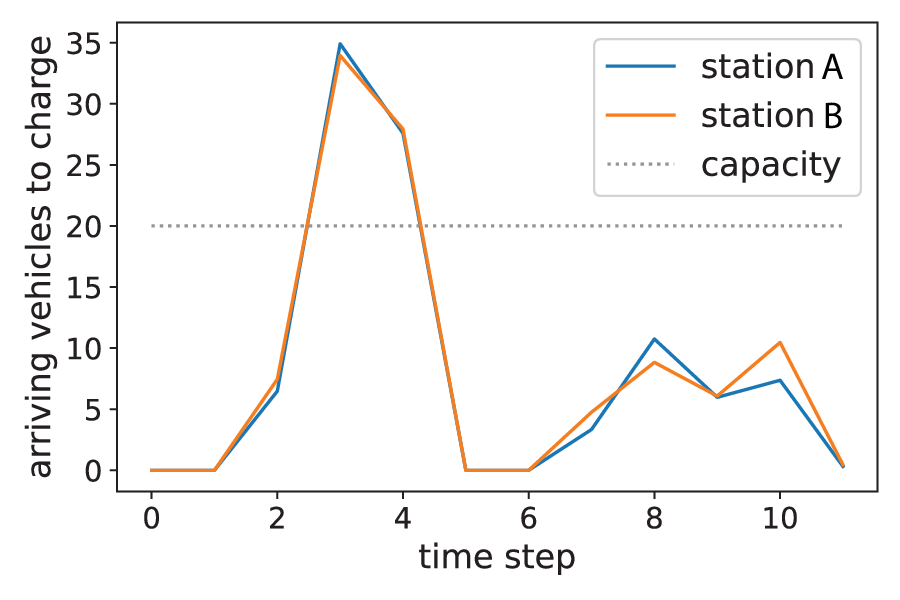}
        \caption{Peak/offpeak}
    \end{subfigure}
    \begin{subfigure}{0.3\textwidth}
        \centering
        \includegraphics[width=0.8\textwidth]{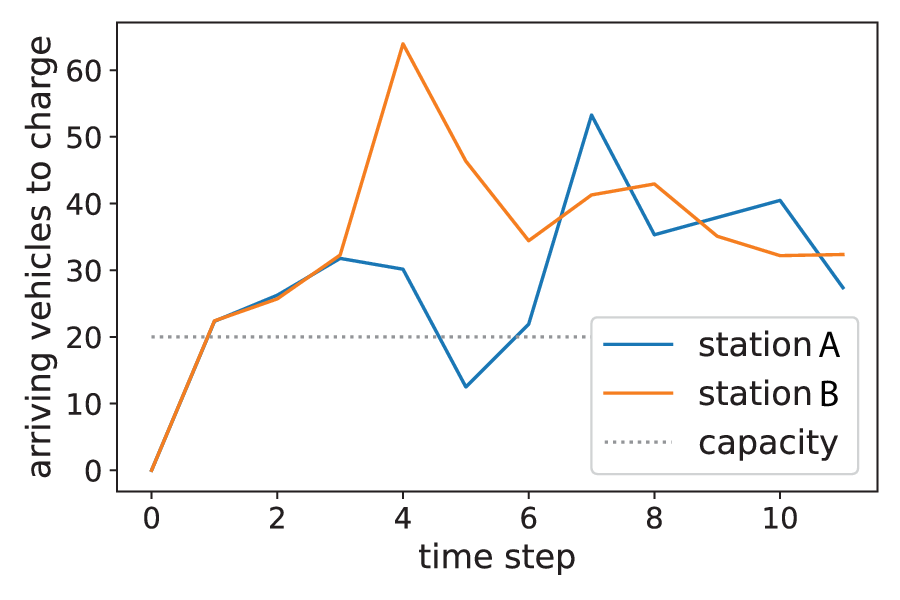}
        \caption{Central/peripheral}
    \end{subfigure}
    \caption{Arriving charging vehicles at inner and outer stations.}
    \label{fig:charge_veh}
\end{figure*}

\begin{figure*}[htb]
    \centering
    \begin{subfigure}{0.3\textwidth}
        \centering
        \includegraphics[width=0.8\textwidth]{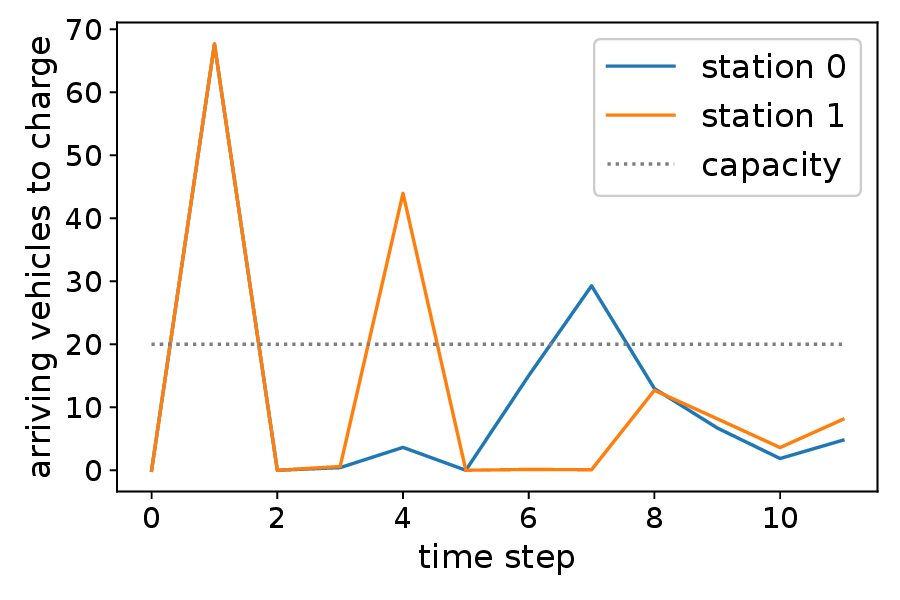}
        \caption{Random initial SOC case 1}
    \end{subfigure}
    \begin{subfigure}{0.3\textwidth}
        \centering
        \includegraphics[width=0.8\textwidth]{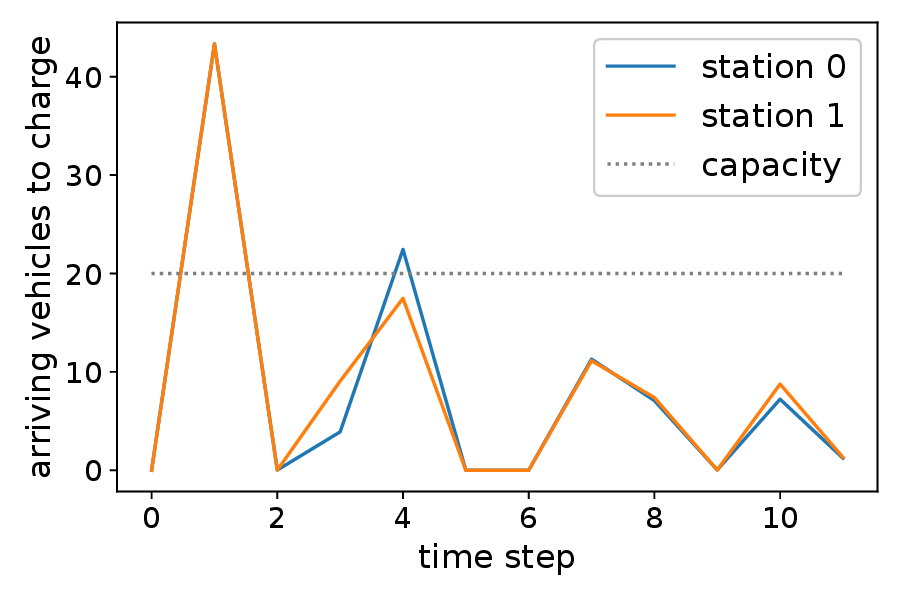}
        \caption{Random initial SOC case 2}
    \end{subfigure}
    \begin{subfigure}{0.3\textwidth}
        \centering
        \includegraphics[width=0.8\textwidth]{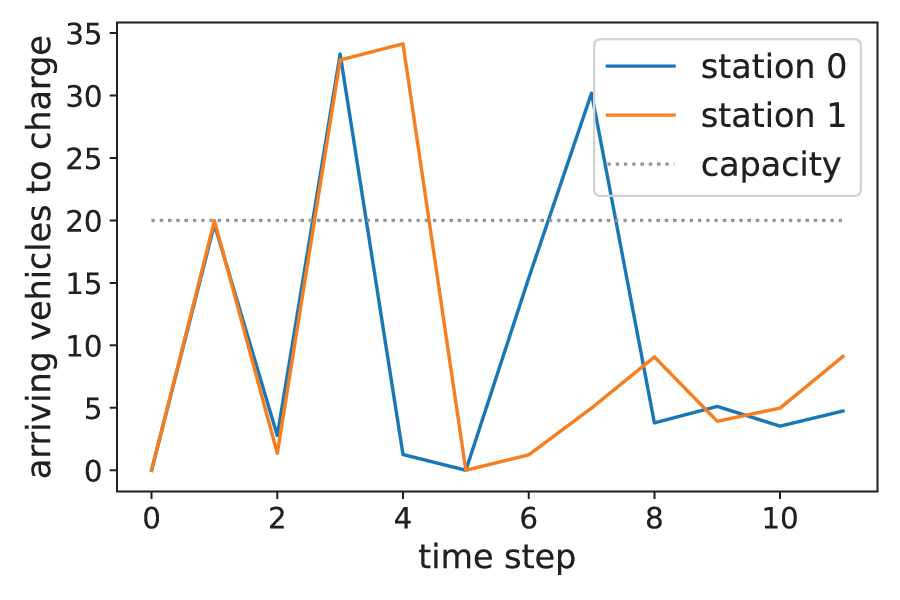}
        \caption{Random initial location}
    \end{subfigure}
    \caption{Arriving charging vehicles at inner and outer stations under different initial state distributions.}
    \label{fig:charge_veh_s0}
\end{figure*}

Thanks to the simple network structure, results in Figure~\ref{fig:congestion} are symmetric. Hence, in what follows, we only compare the vehicle flows between the central zone ($i=6$) and a peripheral zone ($i=0$), and between an inner station ($l=\text{B}$) and an outer station ($l=\text{A}$). 
Figure~\ref{fig:cruise_veh} illustrates the idle vehicles over time along with the demand in each zone represented by the grey dashed line. It can be seen that in the cases of \emph{Uniform} and \emph{Peak/offpeak} demand, a fraction of demand is lost between $t=2$ and $t=4$ because most vehicles are either delivering trips or charging. Yet, this issue is rather minor in \emph{Central/peripheral}. Another observation is that, due to its popularity, the central zone tends to be oversupplied in the second half of the study horizon meanwhile the peripheral zones are undersupplied. This implies that the selfish behaviors of vehicles do not lead to efficient system performance. This is in line with the widely known result in congestion games \cite{roughgarden2007routing}. 

In Figure~\ref{fig:charge_veh}, we plot the arriving vehicles with the station capacity. It clearly illustrates the peak of vehicle arrivals at time $t=3$ and $t=4$, as well as a secondary peak at time $t=7$ in \emph{Uniform} and \emph{Peak/offpeak}. While the first peak is more significant at inner stations, the second one has a more profound impact on outer stations.

\subsection{Impact of initial state distribution}
One possible cause of the severe charging queue at $t=4$ is that we have initiated vehicles with a full battery and thus they all tend to run out of battery around the same time if they continuously operate from the very beginning. Hence, we conduct another experiment with random initial state distributions. Due to the limited space, we only present the results of \emph{Uniform}; the main findings in other cases are similar.

Figure~\ref{fig:charge_veh_s0} presents the arriving vehicle flows in two random samples of initial SOC and one random sample of initial locations. It can be seen that the congestion at the charging stations persists and is more sensitive to initial SOC compared to locations. In the two tested cases, the highest peak happens earlier at time $t=1$, largely because there are more vehicles starting the day with a partial state of charge. Considering that the market has a much larger supply at the beginning (as all vehicles are idle at time $t=0$), these vehicles tend to charge at early time steps.

\section{CONCLUSIONS}
This paper presents a game-theoretical model for the routing and charging behaviors of electric ride-hailing vehicles. Each vehicle's decisions are characterized by a Markov decision process (MDP). 
We show the state transitions in this MDP depend on the aggregate vehicle flows, which are induced by the mean policy over all vehicles. Accordingly, the collective behaviors can be cast in a mean-field Markov game. We define the equilibrium, prove its existence and investigate the equilibrium vehicle flows through numerical experiments. 
The results demonstrate the inefficiency of selfish routing and charging decisions, which causes significant congestion both in service zones and at charging stations. This indicates a necessity to design appropriate control and intervention schemes to optimize the coupled system of on-demand mobility service and electric vehicle charging. 
\revise{Moreover, vehicle travel times within and between zones have assumed to be fixed, which may deviate from the reality. Hence, besides ``congestion'' in the matching and charging processes, future studies shall also take traffic congestion caused by ride-hailing vehicles into consideration.}

\section*{APPENDIX}

\subsection{Specification of meeting probability}
We apply the meeting probability of e-hailing derived in \cite{zhang2023ride} as follows:
\begin{align}
    m^t_i = 
    \begin{cases}
    1-\exp[-\theta^t_{i,1}(\varphi^t_i)^2], & \varphi^t_i > \tilde{\varphi},\\
    1-\exp(-\theta^t_{i,2}\varphi^t_i), & \varphi^t_i \leq \tilde{\varphi},
    \end{cases}
\end{align}
where $\varphi^t_i = q^t_i/y^t_i$ denotes the demand-supply ratio in the local market and $\tilde{\varphi}$ is a threshold value that dictates the oversupplied market condition. The values of  $\theta^t_{i,1}$, $\theta^t_{i,2}$ and $\tilde{\varphi}$ are set according to \cite{zhang2023ride}, which are calibrated from agent-based simulations.

\subsection{Specification of waiting time distribution}\label{appdx:wait-dist}
To specify the probability of waiting time at station $l$ for any vehicle arriving at time $t$, we introduce another intermediate variable $\zeta^t_l(\omega)$ to denote the available charging spots at time $t+\omega$. We again treat $\zeta^t_l\in[0,C]$ as a continuous variable. 
Then, \eqref{eq:meet-prob} is rewritten as
\begin{align}
    w^t_l(\omega) = \min\left(\frac{\zeta^t_l(\omega)}{z^t_l+\varepsilon}, \max\left(0,1-\sum_{\omega'<\omega} w^t_l(\omega')\right)\right),\label{eq:wait-dist-fun}
\end{align}
where $z^t_l = \sum_b z^t_{l,b}$ is the total vehicle arrivals at time $t$, $w^t_l(\omega')=0$ for $\omega'<0$, and $\varepsilon>0$ is an arbitrarily small constant to ensure tractable results. 

Let us first investigate the continuity of \eqref{eq:wait-dist-fun} with respect to vehicle flows $z^t_{l,b}$. Note $\zeta^t_l(\omega)$ is independent of $z^t_{l,b}$ for $\omega=0$ but not necessarily for $\omega>0$ because newly arrived vehicles that finish waiting would take some charging spots. Specifically, after \eqref{eq:wait-dist-fun} is called to compute $w^t_l(\omega)$, $\zeta^t_l(\omega+1)$ will be updated based on the average charging time $\bar{\tau}^t_l = \sum_b z^t_{l,b}\hat{\tau}_b/\sum_b z^t_{l,b}$. The new value of $\zeta^t_l(\omega+1)$ will then be used to compute  $w^t_l(\omega+1)$.
Nevertheless, $\zeta^t_l(\omega)$ remains continuous with $z^t_{l,b}$ as long as $\zeta^t_l(0)$ is continuous in $z^t_{l,b}$, which is guaranteed given the reduced expression $\min\left(\frac{\zeta^t_l(0)}{z^t_l+\varepsilon}, 1\right)$. Therefore, it is safe to conclude that $w^t_l(\omega)$ is continuous in $z^t_{l,b}$, which is numerically demonstrated in Figure~\ref{fig:wait_dist_continuity}.

\begin{figure}[htb]
    \centering
    \includegraphics[width=0.3\textwidth]{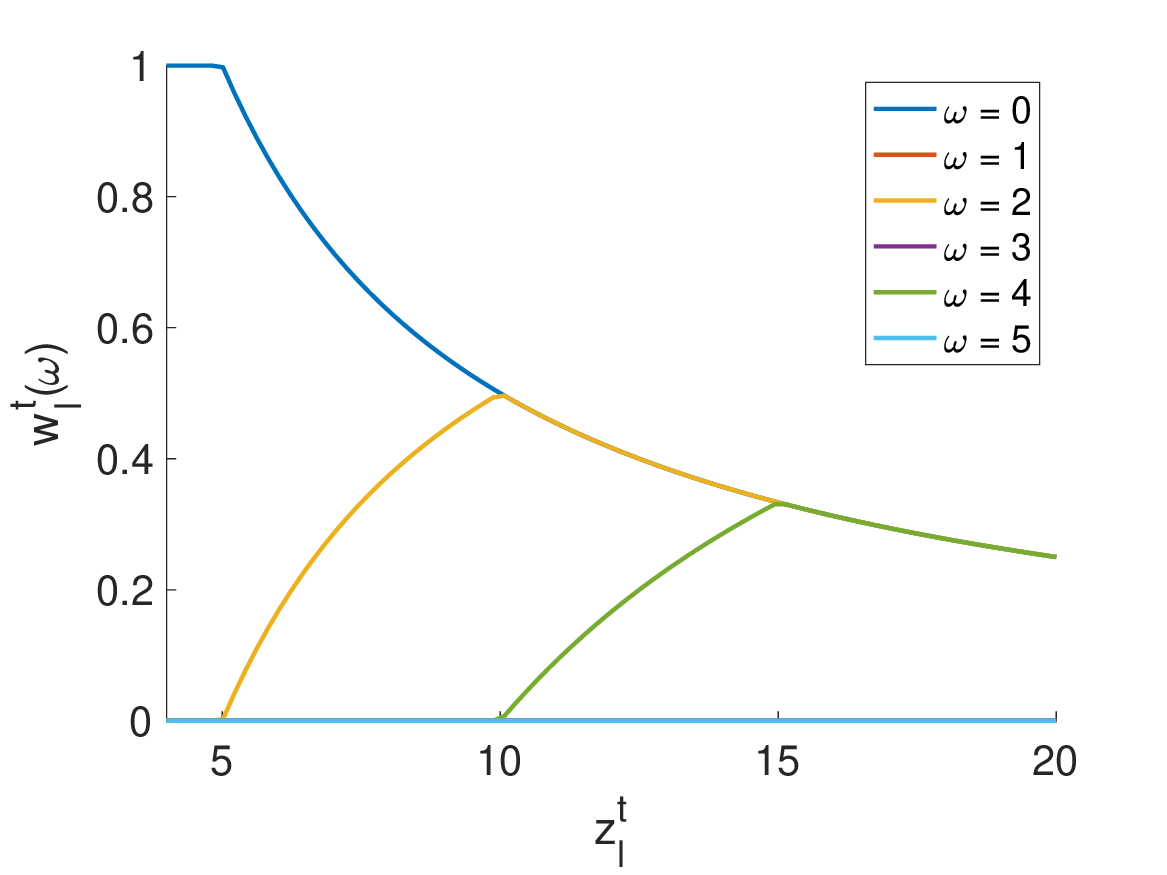}
    \caption{Illustration of continuity of $w^t_l(\omega)$ with respect to $z^t_l$. The initial $\zeta^t_l(\omega)$ is set to be 5 for all time steps; $\bar{\tau}^t_l=1$.}
    \label{fig:wait_dist_continuity}
\end{figure}

Since the available charging spots in the future jointly depend on $z^t_l$, $w^t_l(\omega)$ and $\bar{\tau}^t_l$ in a continuous way, it is easy to show $\zeta^t_l(\omega)$ is continuous in $\mathbf{z}_{\leq t, l}$ by induction. 
This result is numerically demonstrated in Figure~\ref{fig:empty_spot_continuity}.

\begin{figure}[htb]
    \centering
    \includegraphics[width=0.3\textwidth]{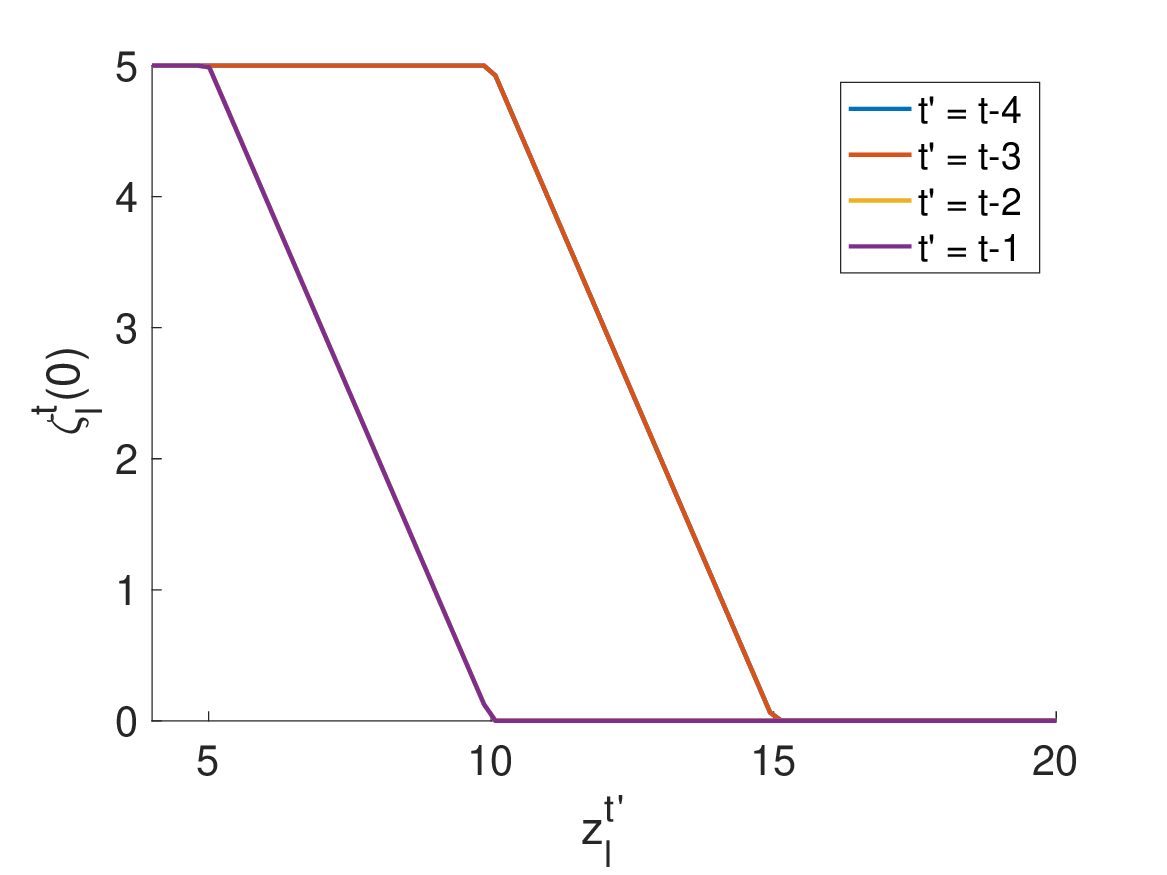}
    \caption{Illustration of continuity of $\zeta^t_l(0)$ with respect to $z^{t'}_l$. The initial $\zeta^t_l(\omega)$ is set to be 5 for all time steps; $\bar{\tau}^t_l=1$; $z^{t''}_l$ is set to zero for $t''\neq t'$.}
    \label{fig:empty_spot_continuity}
\end{figure}

\addtolength{\textheight}{-12cm}   








\bibliographystyle{IEEEtran}
\bibliography{eriver_full}

\end{document}